\documentclass{amsart}
\usepackage{amsmath,amsthm,amssymb}
\usepackage{dsfont}
\usepackage[sort,nocompress,noadjust]{cite}
\usepackage{hyperref}
\newtheorem{theorem}{Theorem}
\newtheorem{cor}{Corollary}
\newtheorem{prop}{Proposition}
\newtheorem{lemma}{Lemma}
\theoremstyle{definition}
\newtheorem{defin}{Definition}

\newtheorem{prob}{Problem}
\usepackage{color}
\usepackage[final]{showlabels}

\newcommand{\cX}{\mathcal{X}}

\newcommand{\RR}{\mathbb{R}}

\newcommand{\1}{\mathds{1}}

\newcommand*{\ep}{\varepsilon}
\newcommand*{\np}{\textsc{np}}
\newcommand*{\rip}{\textsc{rip}}
\newcommand*{\sat}{\textsc{sat}}
\DeclareMathOperator{\val}{val}
\DeclareMathOperator{\poly}{poly}
\DeclareMathOperator*{\argmin}{argmin}

\DeclareMathOperator{\exac}{E1}

\begin{document}

\title{Approximately certifying the restricted isometry property is hard}
\author{Jonathan Weed}
\thanks{This work was supported in part by NSF Graduate Research Fellowship DGE-1122374. The author would like to thank P.\ Rigollet for helpful discussions and the anonymous reviewers for their suggestions and corrections.
\\
Copyright~\copyright~2018 IEEE. Personal use of this material is permitted.  However, permission to use this material for any other purposes must be obtained from the IEEE by sending a request to pubs-permissions@ieee.org.}
\address{Department of Mathematics\\
Massachusetts Institute of Technology\\
77 Massachusetts Avenue\\
Cambridge, MA 02139\\
USA}
\email{jweed@mit.edu}
\date{\today}
\maketitle
\begin{abstract}
A matrix is said to possess the Restricted Isometry Property (RIP) if it acts as an approximate isometry when restricted to sparse vectors. Previous work has shown it to be \np-hard to determine whether a matrix possess this property, but only in a narrow range of parameters. In this work, we show that it is \np-hard to make this determination for \emph{any} accuracy parameter, even when we restrict ourselves to instances which are either RIP or far from being RIP. This result implies that it is \np-hard to approximate the range of parameters for which a matrix possesses the Restricted Isometry Property with accuracy better than some constant. Ours is the first work to prove such a claim without any additional assumptions.
\end{abstract}

\section{Introduction}
The field of compressed sensing, inaugurated by the seminal work of Cand\`es and Tao~\cite{CanTao05} and Donoho~\cite{Don06}, offers an attractive and powerful set of techniques for reconstructing sparse data on the basis of very few measurements.
Implementing compressed sensing techniques in practice involves taking measurements according to a matrix with special properties.
The most widely known such property is the \emph{restricted isometry property}~\cite{Can08}, which requires that the matrix act as an approximate isometry when restricted to sparse vectors.
\begin{defin}\label{def:rip}
A matrix $X \in \RR^{n \times p}$ possesses the \emph{$(k, \delta)$-restricted isometry property} if it satisfies
\begin{equation}\label{eqn:rip_def}
(1-\delta)\|u\|^2 \leq \|X u\|^2 \leq (1+\delta)\|u\|^2
\end{equation}
for all $u \in \RR^p$ with a most $k$ nonzero entries, where $\|\cdot\|$ denotes the $\ell_2$ norm on $\RR^p$.
In this case, we write $X \in \rip(k, \delta)$.
\end{defin}
We call vectors with at most $k$ nonzero entries \emph{$k$-sparse}.
A matrix possessing the $(k, \delta)$-restricted isometry property for appropriate $k$ and $\delta$ can be shown to be a good measurement matrix for compressed sensing.
For example, Cand\`es showed~\cite{Can08} that if $X \in \rip(2k, \delta)$ for $\delta < \sqrt 2 - 1$, then an $\ell_1$ minimization procedure used with the matrix $X$ exactly recovers $k$-sparse vectors.
This result has been sharpened considerably in the intervening years: Cai and Zhang showed~\cite{CaiZha13,CaiZha14} that an $\ell_1$ minimization procedure used with $X$ succeeds in recovering $k$-sparse vectors if $X \in \rip(k, \delta)$ for $\delta < 1/3$ or $X \in \rip(\lceil t k \rceil, \delta)$ for any $t \geq 4/3$ and $\delta < \sqrt{(t-1)/t}$.

Finding \rip{} matrices in the most interesting range of parameters is an object of active study.
Constructing such matrices deterministically is a hard problem~\cite{BanFic13,Dev07,BouDilFor11}, but there are several very simple random methods known to generate \rip{} matrices with high probability~\cite{BarDav08,HavReg16}.
The fact that these randomized algorithms have a small probability of failure motivates the question of certifying whether a given matrix is \rip:
\begin{prob}\label{prob:no_gap}
Given a matrix $X$, a positive integer $k$, and $\delta > 0$, is $X \in \rip(k, \delta)$?
\end{prob}
While previous work has shown Problem~\ref{prob:no_gap} to be computationally hard, these works leave open the question of whether Problem~\ref{prob:no_gap} is \np-hard for a range of parameters which are relevant to practical applications.
(See Section~\ref{sec:prior_work} for a fuller account.)
In particular, earlier work has only been able to show the \np-hardness of deciding whether $X \in \rip(k, \delta)$ for $\delta = 1 - o(1)$.
By contrast, the correct question in practice is to decide whether $X \in \rip(k, \delta)$ for some \emph{constant} $\delta$.

Earlier work has also largely focused on the exact question of being able to distinguish the two alternatives $X \in \rip(k, \delta)$ and $X \notin \rip(k, \delta)$, which in particular requires being able to make this determination even for matrices that are not in $\rip(k, \delta)$ but are nonetheless very close to being so.
In practice, this question may be unnecessarily precise, and it may suffice to know the \emph{approximate} range of parameters for which a matrix possesses the restricted isometry property and thereby circumvent the problem of testing membership in $\rip(k, \delta)$ for matrices near the boundary of the set.

If we wish to find a matrix $X \in \rip(k, \delta)$ and are willing to accept matrices for which this condition holds approximately, a modest goal would be to be able to check whether $X \in \rip(k', \delta')$ for some much weaker set of parameters $k' \ll k$ and $\delta' \gg \delta$.
If $X \notin \rip(k', \delta')$, then $X$ is ``very far'' from being in $\rip(k, \delta)$ and so can be confidently discarded.
Designing such a procedure merely requires being able to tell the difference between matrices in $\rip(k, \delta)$ and matrices that are far from being in $\rip(k, \delta)$, in the sense that they are not even in $\rip(k', \delta')$.
Since the matrices in these two classes are very different, we might hope that distinguishing between these cases is possible even if Problem~\ref{prob:no_gap} is computationally hard.

We formalize this idea by proposing the following modification of Problem~\ref{prob:no_gap}:
\begin{prob}\label{prob:modified_no_gap}
Fix constants $\delta \in (0, 1)$ and $\lambda_1, \lambda_2 > 1$.
Given a matrix $X$ and positive integer $k$, distinguish $X \in \rip(k,\delta)$ from $X \notin \rip(k/\lambda_1, \lambda_2 \delta)$.
\end{prob}
Note that we do not require our procedure to do anything in particular when presented with a matrix for which neither case holds.
Equivalently, we could assume that we are promised \emph{a priori} that the matrix under consideration is either in $\rip(k,\delta)$ or not in $\rip(k/\lambda_1, \lambda_2 \delta)$.
Problem~\ref{prob:modified_no_gap} is therefore easier than Problem~\ref{prob:no_gap}, since all that we seek is a procedure to distinguish between two very different situations.
Problem~\ref{prob:modified_no_gap} is known as a \emph{gap problem} in the computational complexity literature; more details about problems of this type appear in Section~\ref{sec:gap}.

In this work, we show that for all $\delta \in (0, 1)$ there exist $\lambda_1, \lambda_2 >1$ such that Problem~\ref{prob:modified_no_gap} is \np-hard.
This immediately establishes the \np-hardness of Problem~\ref{prob:no_gap} as well, since it implies that even if $X \in \rip(k, \delta)$, we nevertheless cannot efficiently certify that it satisfies the weaker condition $X \in \rip(k/\lambda_1, \lambda_2 \delta)$.

We also consider the following two related problems.

\begin{prob}\label{prob:approximate_k}
Given a matrix $X$ and $\delta > 0$, find the largest positive integer $k$ such that $X \in \rip(k, \delta)$.
\end{prob}
\begin{prob}\label{prob:approximate_delta}
Given a matrix $X$ and positive integer $k$, find the smallest $\delta > 0$ such that $X \in \rip(k, \delta)$.
\end{prob}

Our results show that Problems~\ref{prob:approximate_k} and~\ref{prob:approximate_delta} are
hard to solve even approximately.

\subsection{Gap hardness}\label{sec:gap}
Gap problems are part of a broader class of \emph{promise problems}, where the input is guaranteed to fall into one of two classes.
In a gap problem, these two classes are assumed to be well separated.
Concretely, given a general optimization problem of the form
\begin{equation}\label{eqn:gap_problem}
\max_{x \in \cX} f(x)\,,
\end{equation}
a threshold $t$, and a constant $c > 1$, the $c$-gap problem is to distinguish between
\begin{equation*}
\max_{x \in \cX} f(x) < t/c \text{ and } \max_{x \in \cX} f(x) \geq t\,,
\end{equation*}
where we have been promised that one of the two situations holds.
If this gap problem is \np-hard, then it is clearly \np-hard to find a $c$-multiplicative approximation to~\eqref{eqn:gap_problem}.
For this reason, gap hardness results immediately imply that the corresponding approximation problem is also hard~\cite{Vaz13}.

To establish the \np-hardness of Problem~\ref{prob:modified_no_gap}, we produce a reduction from an \np-hard gap problem and show that the reduction preserves the gap between classes. We reduce from a problem known as ``max positive 1-in-3 \sat.''

\begin{defin}\label{def:e1}
Given boolean variables $x_1, \dots, x_k$, the predicate \emph{exactly one} ($\exac$) is given by
\begin{equation*}
\exac(x_1, \dots, x_k) = \left\{\begin{array}{ll}
T & \text{ if exactly one of $x_1$, \dots, $x_k$ is true} \\
F & \text{ otherwise.}
\end{array}\right.
\end{equation*}

A \emph{max positive 1-in-3 \sat} instance is a set of $\exac$ clauses $\{c_1, \dots, c_m\}$ each containing at most $3$ variables from a set $\{x_1, \dots, x_n\}$.
The word ``positive'' is used to stress that all variables appear in positive form.
If $\exac(c_i) = T$, we say that clause $i$ is \emph{satisfied}.
In this work, we will consider instances with a further restriction.
\begin{defin}
A instance of max positive 1-in-3 \sat\ is \emph{$6$-bounded} if each variable appears in at most six clauses.
\end{defin}

Given an instance $\phi$ of $6$-bounded max positive 1-in-3 \sat\ and an assignment $\mathbf{x} \in \{T, F\}^n$, denote by $\val(\phi, \mathbf{x})$ the proportion of clauses satisfied by the assignment $\mathbf{x}$.
Finally, define
\begin{equation*}
\val(\phi) := \max_{\mathbf{x} \in \{T, F\}^n} \val(\phi, \mathbf{x})\,.
\end{equation*}
If $\val(\phi) = 1$, that is, if there is an assignment satisfying all clauses, we say that $\phi$ is \emph{satisfiable}.
\end{defin}

Our reduction is based on the following proposition.
\begin{prop}\label{prop:max_bounded_monotone_hardness}
There exists a constant $\alpha$ such that, given a $6$-bounded max positive 1-in-3 \sat\ instance $\phi$, it is \np-hard to distinguish between $\val(\phi) = 1$ and $\val(\phi) < (1-\alpha)$.
Moreover, the instances $\phi$ under consideration can be restricted to contain exactly $9/13$ as many clauses as variables.
\end{prop}

A proof of Proposition~\ref{prop:max_bounded_monotone_hardness} appears in the Appendix.
It is well known that deciding whether an instance $\phi$ is \emph{satisfiable} is an \np-complete problem~\cite{Sch78}, and the gap hardness of max positive 1-in-3 \sat\ (without the $6$-boundedness condition) is proved in~\cite{KhaSud96}.
Hardness problems of this type were first officially stated in~\cite{PapYan91}, and their \np-hardness follows from the celebrated \textsc{pcp} Theorem~\cite{AroLunMot98}.

\subsection{Main result}\label{subsec:main}
We show the following gap hardness result for Problem~\ref{prob:modified_no_gap}.
\begin{theorem}\label{thm:main}
For all $\delta \in (0, 1)$ there exist constants $\lambda_1, \lambda_2 > 1$ such that, given a matrix $X$ and sparsity parameter $k$, it is \np-hard to decide whether $X \in \rip(k, \delta)$ or $X \not\in \rip(k/\lambda_1, \lambda_2 \delta)$.  Moreover, the claim holds even when restricted to matrices $X$ satisfying $\|Xu\|^2 \leq (1+\delta)\|u\|^2$ for all $u$.
\end{theorem}
For $\lambda_1, \lambda_2 > 1$, the condition that $X \in \rip(k/\lambda_1, \lambda_2 \delta)$ is weaker than $X \in \rip(k, \delta)$, since the bounds in Equation~\eqref{eqn:rip_def} are weaker and the sparsity condition is stronger, so that Equation~\eqref{eqn:rip_def} is required to hold for a smaller set of vectors.
Theorem~\ref{thm:main} says that even if $X$ satisfies the strong condition $X \in \rip(k, \delta)$, it is hard to even certify that it satisfies the weak condition $X \in \rip(k/\lambda_1, \lambda_2 \delta)$. The restriction to $X$ such that $\|Xu\|^2 \leq (1+\delta)\|u\|^2$ implies that Proplem~\ref{prob:modified_no_gap} is hard even when only the \emph{lower} bound of~\eqref{eqn:rip_def} is in question.
We focus on this case because the lower bound of~\eqref{eqn:rip_def} is more important in compressed sensing applications~\cite[Remark~1]{BlaCarTan11}.

Theorem~\ref{thm:main} implies the following hardness of approximation results.

\begin{cor}\label{cor:approximate_k}
There exists a constant $\lambda_1> 1$ such that it is \np-hard to solve Problem~\ref{prob:approximate_k} to within a $\lambda_1$ factor.
\end{cor}
\begin{proof}
Fix an arbitrary $\delta \in (0, 1)$, and let $\lambda_1$ and $\lambda_2$ be the corresponding constants appearing in the statement of Theorem~\ref{thm:main}.
Given a matrix $X$ and desired sparsity parameter $k$, let $k^*$ be the largest positive integer such that $X \in \rip(k^*, \lambda_2 \delta)$.
If $X \in \rip(k, \delta) \subseteq \rip(k, \lambda_2 \delta)$, then $k^* \geq k$.
If $X \notin \rip(k/\lambda_1, \lambda_2 \delta)$, then $k^* < k/\lambda_1$.
A procedure to find $k'$ such that $k' \in [k^*/\lambda_1, k^*]$ would therefore yield a procedure to distinguish $X \in \rip(k, \delta)$ from $X \notin \rip(k/\lambda_1, \lambda_2 \delta)$, which is \np-hard by Theorem~\ref{thm:main}.
\end{proof}

\begin{cor}\label{cor:approximate_delta}
There exists a constant $\lambda_2 > 1$ such that it is \np-hard to solve Problem~\ref{prob:approximate_delta} to within a $\lambda_2$ factor.
\end{cor}
\begin{proof}
The proof is identical to the proof of Corollary~\ref{cor:approximate_k}.
Fix an arbitrary $\delta \in (0, 1)$, and let $\lambda_1$ and $\lambda_2$ be the corresponding constants appearing in the statement of Theorem~\ref{thm:main}.
Given a matrix $X$ and desired sparsity parameter $k$, let $\delta^*$ be the smallest positive constant such that $X \in \rip(k/\lambda_1, \delta^*)$.
If $X \in \rip(k, \delta) \subseteq \rip(k/\lambda_1, \delta)$, then $\delta^* \leq \delta$.
If $X \notin \rip(k/\lambda_1, \lambda_2 \delta)$, then $\delta^* > \lambda_2 \delta$.
A procedure to find $\delta'$ such that $\delta' \in [\delta^*, \lambda_2 \delta^*]$ would yield a procedure to distinguish $X \in \rip(k, \delta)$ from $X \notin \rip(k/\lambda_1, \lambda_2 \delta)$, which is \np-hard by Theorem~\ref{thm:main}.
\end{proof}

\subsection{Proof strategy}\label{sec:proof_strategy}
Suppose we have an instance $\phi$ of $6$-bounded positive 1-in-3 \sat\ with $n$ variables and $m$ clauses.
Given such an instance, define the $m \times n$ matrix $\Phi$:
\begin{equation}\label{eq:phi_def}
\Phi_{ij} = \left\{\begin{array}{ll}
1 & \text{if variable $j$ appears in clause $i$,} \\
0 & \text{otherwise.}\end{array}\right.
\end{equation}
Any vector $v \in \{0, 1\}^n$ can be interpreted as an assignment of true and false to $n$ variables, where $v_j = 1$ if variable $j$ is true, and $v_j = 0$ otherwise.
Given such a vector, the definition of $\Phi$ implies
\begin{align*}
(\Phi v)_i =1 & \iff \text{clause $i$ contains exactly one true variable} \\ 
&\iff \text{clause $i$ is satisfied}\,.
\end{align*}

We obtain that $\phi$ is satisfiable if and only if there exists a $0$-$1$ vector $v$ such that $\Phi v = \1$, the all-ones vector.
On the other hand, if $\val(\phi) < 1 - \alpha$, then $\|\Phi v - \1\|^2 > \alpha m$ for all $v \in \{0, 1\}^n$.
In summary,
\begin{align*}
\val(\phi) = 1 & \iff \min_{v \in \{0, 1\}^n} \|\Phi v - \1\|^2 = 0 \\
\val(\phi) < 1 - \alpha & \iff \min_{v \in \{0, 1\}^n} \|\Phi v - \1\|^2 > \alpha m\,.
\end{align*}
In other words, being able to compute
\begin{equation}\label{eqn:phi_program}
\min_{v \in \{0, 1\}^n} \|\Phi v - \1\|^2
\end{equation}
would immediately yield a procedure to check whether $\val(\phi) = 1$ or $\val(\phi) < 1 - \alpha$.
Under the assumption that computing $\val(\phi)$ is intractable, we obtain that the problem in~\eqref{eqn:phi_program} must also be hard to solve.
Moreover, the gap hardness of computing $\val(\phi)$ implies that finding an approximation to~\eqref{eqn:phi_program} is \np-hard.

Given a vector $u$, denote by $\|u\|_0$ the number of nonzero entries of $u$.
We will construct a matrix $X$ and sparsity parameter $k$ such that the value of
\begin{equation}\label{eqn:X_program}
\min_{u: \|u\| = 1, \|u\|_0 \leq k} \|X u\|
\end{equation}
is approximately the same as the value of~\eqref{eqn:phi_program}, up to an additive shift.

The matrix $X$ we construct will contain a rescaled version of $\Phi$ as a submatrix.
The remaining entries of $X$ will be chosen in such a way to ensure that the sparse vectors $u$ for which $\|X u\|^2$ is minimized are approximately $0$-$1$ vectors, and hence correspond approximately to feasible vectors $v$ in~\eqref{eqn:phi_program}.
Then, we will argue that for $0$-$1$ vectors, the values of~\eqref{eqn:phi_program} and~\eqref{eqn:X_program} are equal up to an additive shift.
By carefully controlling the errors at every step, we show knowledge of the value of~\eqref{eqn:X_program} up to some constant level of accuracy would imply the ability to solve~\eqref{eqn:phi_program}, and hence the ability to estimate $\val(\phi)$.
The hardness of the latter program then completes the proof.

\subsection{Prior work}\label{sec:prior_work}
Several papers have shown Problem~\ref{prob:no_gap} to be computationally intractable under a number of different assumptions~\cite{NatWu14,KoiZou14,BanDob13,TilPfe14,WanBerPla16}.

In~\cite{KoiZou14}, the authors analyze a problem similar to our Problem~\ref{prob:modified_no_gap}.
They obtain a variety of hardness results for the problem of distinguishing between $X \in \rip(k, \delta)$ and $X \notin \rip(k', \delta')$ under a variety of assumptions, which are plausible but nevertheless stronger than $\textsc{p} \neq \np$.
Their results suggest that distinguishing between $X \in \rip(k, \delta)$ and $X \notin \rip(k', \delta')$ is computationally hard, at least when $|\delta - \delta'|$ approaches zero as the size of the instance increases.

Another line of work has succeeded in showing that \rip\ certification is \np-hard without requiring additional assumptions.
The first two papers to prove the \np-hardness of Problem~\ref{prob:no_gap}~\cite{BanDob13,TilPfe14} both rely on the fact that given a matrix $X \in \RR^{n \times p}$ and a sparsity parameter $k$, it is \np-hard to certify whether the kernel of $X$ contains a nonzero $k$-sparse vector.
When no such vector exists, one can show
\begin{equation*}
\|X u\|^2 \geq 2^{- \poly(n, p)}\|u\|^2 \quad \text{for all $k$-sparse $u$}\,.
\end{equation*}
These reductions therefore show that certifying $X \in \rip(k, \delta)$ is hard when $\delta = 1 - \ep$ for some $\ep$ that is exponentially small in $n$ and $p$.

The above results establish that it is \np-hard to determine whether $X \in \rip(k, \delta)$ even when only the \emph{lower} bound of~\eqref{eqn:rip_def} is in question.
As noted in Section~\ref{subsec:main}, it is the lower bound that is more important for compressed sensing applications.
However,~\cite{TilPfe14} also show that certifying that a matrix $X$ satisfies the upper bound in~\eqref{eqn:rip_def} is strongly \np-hard via a separate reduction from the Clique problem.
Their reduction establishes that it is hard to distinguish between the case that
\begin{equation*}
\|Xu\|^2 = (n^2 + k - 1)\|u\|^2 \quad \text{for some $k$-sparse $u$}
\end{equation*}
and
\begin{equation*}
\|Xu\|^2 \leq (n^2 + k - 1 - o(1))\|u\|^2 \quad \text{for all $k$-sparse $u$}\,.
\end{equation*}

The work most similar to ours is~\cite{NatWu14}, in which the authors raised the same objections we do about the restrictiveness of Problem~\ref{prob:no_gap}.
They prove that Problem~\ref{prob:modified_no_gap} is hard under the \emph{small-set expansion hypothesis} (see~\cite{NatWu14} for a definition), which asserts that a particular graph problem is \np-hard to approximate.
Our work establishes this result without requiring this extra hypothesis.

One very interesting direction of recent work has de-emphasized the worst-case nature of the above results.
Instead, motivated by the fact that many random constructions are known to generate matrices possessing the restricted isometry property with high probability~\cite{BarDav08,HavReg16}, the paper~\cite{WanBerPla16} considered the important question of whether Problem~\ref{prob:no_gap} is hard \emph{on average}.
Their proof establishes that Problem~\ref{prob:no_gap} is hard for a natural random model under a hypothesis known as the \emph{planted dense subgraph assumption}.
Whether their techniques can be extended to Problems~\ref{prob:modified_no_gap}--\ref{prob:approximate_delta} is an open question.

\section{Proof of main theorem}
In this section, we prove Theorem~\ref{thm:main}, which we recall below.
\begingroup
\def\thetheorem{\ref{thm:main}}
\begin{theorem}
For all $\delta \in (0, 1)$ there exist constants $\lambda_1, \lambda_2 > 1$ such that, given a matrix $X$ and sparsity parameter $k$, it is \np-hard to decide whether $X \in \rip(k, \delta)$ or $X \not\in \rip(k/\lambda_1, \lambda_2 \delta)$. Moreover, the claim holds even when restricted to matrices $X$ satisfying $\|Xu\|^2 \leq (1+\delta)\|u\|^2$ for all $u$.
\end{theorem}
\endgroup

We first prove the statement for a specific choice of $\delta$, and later show how the proof can be extended to all $\delta \in (0, 1)$.
We note that while the matrices constructed in the main reduction of Theorem~\ref{thm:main} have more rows than columns, we show in Section~\ref{sec:shape} how to extend the results to square matrices and matrices with more columns than rows, which are the shapes relevant for compressed sensing applications.

\subsection{Proof overview}\label{sec:overview}
The reduction is from 6-bounded positive 1-in-3 \sat.
By Proposition~\ref{prop:max_bounded_monotone_hardness}, there exists a constant $\alpha$ such that it is \np-hard to distinguish satisfiable 6-bounded positive 1-in-3 \sat\ instances from instances in which only a $1-\alpha$ fraction of clauses are satisfiable.

We show that there exist three positive constants $c_1$, $c_2$, and $c_3$ with $c_2 > 1$ such that, given $\phi$ with $n$ variables and $m$ clauses, we can construct a matrix $\tilde X \in \RR^{(4n + m) \times 3n}$ with the following three properties:
\begin{enumerate}
\item The matrix $\tilde X$ has operator norm at most $c_1$.
\item If $\val(\phi) = 1$, then there exists a vector $u$ with $\|u\|_0 = 2n$ such that
\begin{equation*}
\|\tilde Xu\|^2 = \frac 1 2 \|u\|^2\,.
\end{equation*}
\item If $\val(\phi) < 1 - \alpha$, then every vector $u$ with $\|u\|_0 \leq 2 c_2 n$ satisfies
\begin{equation*}
\|\tilde Xu\|^2 \geq \frac{1+c_3}{2}\|u\|^2\,.
\end{equation*}
\end{enumerate}

Given $\tilde X$ with the above three properties, consider the matrix $X = \frac{1}{c_1} \tilde X$.
By Property~1, $\|Xu\|^2 \leq \|u\|^2$ for all vectors $u$.
Choose $\delta = 1 - \frac{1+c_3}{2c_1^2}$ and $\lambda_2 = \frac{2c_1^2}{2c_1^2-c_3}$.
We obtain the following: if $\val(\phi) < 1-\alpha$, then for all $u$ such that $\|u\|_0 \leq 2c_2 n$,
\begin{equation*}
\|Xu\|^2 \geq \frac{1 + c_3}{2c_1^2} \|u\|^2 = (1-\delta)\|u\|^2\,.
\end{equation*}
Conversely, if $\val(\phi) = 1$ then there exists a $u$ satisfying $\|u\|_0 = 2n$ such that
\begin{equation*}
\|Xu\|^2 = \frac{1}{2c_1^2} \|u\|^2< (1-\lambda_2 \delta)\|u\|^2
\end{equation*}
In other words, letting $k = 2c_2n$ and $\lambda_1 = c_2$ yields
\begin{align*}
\val(\phi) < 1-\alpha & \implies X \in \rip(k, \delta) \\
\val(\phi) = 1 & \implies X \notin \rip(k/\lambda_1, \lambda_2 \delta)\,.
\end{align*}
Since it is \np-hard to distinguish between $\val(\phi) = 1$ and $\val(\phi) < 1- \alpha$, is is also \np-hard to distinguish between $X \in \rip(k, \delta)$ and $X \notin \rip(k/\lambda_1, \lambda_2\delta)$.

The matrix $\tilde X$ is defined in Section~\ref{sec:X_def}. We first verify Properties~1 and~2, and then in Section~\ref{sec:prop_2} reduce the verification of Property~3 to verifying two conditions on the minimizer of the program given in~\eqref{eqn:the_program}. We verify these conditions in Section~\ref{sec:verify}. Finally, we show how to extend the proof to general $\delta$ in Section~\ref{sec:extension}, and to matrices of other shapes in Section~\ref{sec:shape}.
\subsection{Definition of $\tilde X$}\label{sec:X_def}
Let $\ep$ and $\xi$ be small rational constants to be chosen later, with $\xi \ll \ep < 1$.

Given a 6-bounded 1-in-3 \sat\ instance $\phi$ with $n$ variables and $m$ clauses, define $\Phi \in \RR^{m \times n}$ as in~\eqref{eq:phi_def}.
Let $I$ be the $n \times n$ identity matrix and $\1$ the all-ones vector of length $n$, and define
\begin{equation*}
P = I - \frac{1}{n}\1 \1^\top\,.
\end{equation*}
$P$ is an orthogonal projection onto the subspace orthogonal to $\1$.
Let $\tilde X \in \RR^{(4n+m) \times 3n}$ be the following matrix:
\begin{equation}\label{eqn:x_def}
\tilde X = \left(\begin{array}{c|c|c}
I & 0 & 0 \\ \hline
0 & I & 0 \\ \hline
0 & 0 & \xi^{-1} P \\ \hline
\xi^{-1} I & \xi^{-1} I & - \xi^{-1} I \\ \hline
\ep \Phi & 0 & - \ep I'
\end{array}\right)\,,
\end{equation}
where $I'$ is the $n \times n$ identity matrix truncated to have only $m$ rows.
The entries of this matrix are rational constants with bit complexity independent of the size of $X$, so $\tilde X$ can be constructed in polynomial time given the input instance $\phi$.

Call a vector $u \in \{0, 1\}^{3n}$ an \emph{assignment vector} if
\begin{alignat*}{2}
u_i + u_{i+n} & = 1 && \quad \text{ for $1 \leq i \leq n$,} \\
u_j & = 1 && \quad \text{ for $2n < j \leq 3n$.}
\end{alignat*}
We can interpret such vectors as true-false assignments to $n$ variables by setting $x_i = T$ if $u_i = 1$ and $x_i = F$ if $u_{i+n} = 1$.
\begin{prop}\label{prop:assignment_value}
If $u$ is an assignment vector, then $\|u\|^2 = \|u\|_0 = 2n$ and
\begin{equation*}
n + \ep^2 \alpha m \leq \|\tilde Xu\|^2 \leq n + 4 \ep^2 \alpha m\,,
\end{equation*}
where $\alpha$ is the proportion of clauses in $\phi$ not satisfied by the true-false assignment corresponding to $u$. Moreover, if no clause in $\phi$ contains three true variables, then the lower bound holds with equality.
\end{prop}
\begin{proof}
To evaluate $\tilde Xu$, we write
\begin{equation*}
\tilde Xu = \left(\begin{array}{c}
y_1 \\ \hline
y_2 \\ \hline
y_3 \\ \hline
y_4 \\ \hline
y_5
\end{array}\right)\,,
\end{equation*}
where $y_1, \dots, y_4 \in \RR^n$ and $y_5 \in \RR^m$.
Clearly $\|y_1\|^2 + \|y_2\|^2 = n$.
Since $P \1 = 0$ by definition, $y_3 = 0$.
Likewise, $y_4 = 0$ because $u_i + u_{i+n} = 1$ for $1 \leq i \leq n$.
Write $u^+ \in \RR^n$ for the vector consisting of the first $n$ coordinates of $u$.
By definition of $\Phi$, we have
\begin{equation*}
\|y_5\|^2 = \ep^2\|\Phi u^+ - \1\|^2\,.
\end{equation*}
If the $j$th clause of $\phi$ is satisfied by the assignment corresponding to $u$, then $(\Phi u^+)_j = 1$; otherwise, $1 \leq |(\Phi u^+)_j - 1| \leq 2$.
Therefore
\begin{equation*}
\alpha m \leq \|\Phi u^+ - \1\|^2 \leq 4 \alpha m\,.
\end{equation*}
If no clause in $\phi$ contains three true variables, then $|(\Phi u^+)_j - 1| \leq 1$ for all $j \in [m]$, so the lower bound holds with equality.
\end{proof}

With this choice of $\tilde X$, Properties 1 and 2 are easy to establish.

\begin{prop}\label{prop:property_1}
If $\val(\phi) = 1$, then there exists a $u \in \RR^{3n}$ such that $\|u\|_0 = 2n$ and
\begin{equation*}
\|\tilde Xu\|^2 = \frac 1 2\|u\|^2\,.
\end{equation*}
\end{prop}
\begin{proof}
It suffices to produce a vector $u$ such that $\|u\|^2 = \|u\|_0 = 2n$ and
\begin{equation*}
\|\tilde Xu\|^2 = n\,.
\end{equation*}
Let $u$ be the assignment vector corresponding to a satisfying assignment of $\phi$.
Applying Proposition~\ref{prop:assignment_value} yields the claim.
\end{proof}

\begin{prop}\label{prop:property_3}
The matrix $\tilde X$ defined in Equation~\eqref{eqn:x_def} has operator norm at most $3\xi^{-1}$.
\end{prop}
\begin{proof}
We employ the following upper bound on the size of the largest singular value due to Schur~\cite{Sch11}, which is well known (see, e.g.,~\cite{Gol13}).
If $r_i$ is the $\ell_1$ norm of the $i$th row and $c_j$ the $\ell_1$ norm of the $j$th column, then
\begin{equation*}
\|\tilde X\|^2_{\text{op}} \leq \max_{i, j} r_i c_j\,.
\end{equation*}
It is then easy to check that $r_i \leq 3 \xi^{-1}$ and $c_j \leq 3 \xi^{-1}$.
The claim follows.
\end{proof}

\subsection{Proof of Property~3}\label{sec:prop_2}
The remainder of the proof is dedicated to showing that Property~3 holds with $c_2 = 1 + \xi^2$ and $c_3$ to be specified.
For simplicity, we consider vectors $u$ satisfying $\|u\|^2 = 2n$.
In what follows, let
\begin{equation}\label{eqn:the_program}
w \in \argmin_{u: \|u\|^2 = 2n, \|u\|_0 \leq 2(1+\xi^2)n} \|\tilde Xu\|^2\,.
\end{equation}
Since $w$ and $-w$ are both minimizers, we assume without loss of generality that $w$ is such that the average value of the last $n$ entries is nonnegative.

We aim to show that, if $\val(\phi) < 1-\alpha$, then
\begin{equation}\label{eq:goal}
\|\tilde Xw\|^2 \geq (1 + c_3) n
\end{equation}
for some constant $c_3$.

Proposition~\ref{prop:assignment_value} implies that the value of $\|\tilde Xu\|$ for an assignment vector is directly related to the number of satisfied clauses in the true-false assignment corresponding to $u$.
To show~\eqref{eq:goal}, we will argue that $w$ is ``approximately'' an assignment vector, so that $\|\tilde Xw\|$ can still be controlled by $\val(\phi)$.
We interpret $w \in \RR^{3n}$ as the concatenation of three vectors $w^+$, $w^-$, and $v$ in $\RR^n$.
To show that $w$ is approximately an assignment vector, we need to show that $v$ is close to the all-ones vector, that $w^+ + w^-$ is also close to the all-ones vector, and that $w^+$ and $w^-$ have almost disjoint support.

Call variable $i$ \emph{good} if exactly one of $w^+_i$ and $w^-_i$ is zero, and the other lies in the interval $(2/3, 4/3)$.
Call clause $j$ \emph{good} if all the variables it contains are good and $v_j$ lies in the interval $(5/6, 7/6)$.
Call a clause \emph{bad} if it is not good.

\begin{prop}\label{prop:approximate_assignment_implies_bound}
Let $w$ be a minimizer in~\eqref{eqn:the_program}.
Suppose that there exist positive constants $\beta$ and $\gamma$ such that the following two properties hold:
\begin{itemize}
\item  $\|w^+\|^2 + \|w^-\|^2 \geq (1 - \beta)n$
\item There are at most $\gamma n$ bad clauses.
\end{itemize}
Let
\begin{equation*}
\rho = \frac{\ep^2}{36}\left(\frac{9}{13} \alpha - \gamma\right) - \beta \,.
\end{equation*}
If $\rho > 0$, then Property~3 holds with $c_2 = 1+\xi^2$ and $c_3 = \rho$.
\end{prop}
\begin{proof}
As in the proof of Proposition~\ref{prop:assignment_value}, write
\begin{equation*}
\tilde Xw = \left(\begin{array}{c}
y_1 \\ \hline
y_2 \\ \hline
y_3 \\ \hline
y_4 \\ \hline
y_5
\end{array}\right)\,,
\end{equation*}
Then
\begin{equation*}
\|y_1\|^2 + \|y_2\|^2 = \|w^+\|^2 + \|w^-\|^2 \geq (1 - \beta)n\,.
\end{equation*}
Denote by $\varphi$ the $1$-in-$3$ \sat\ instance consisting only of the good clauses in $\phi$.
The vector $w$ induces a true-false assignment to the variables in $\phi$ in the following way: if the $i$th variable appears in $\varphi$, then it is good, so exactly one of $w^+_i$ and $w^-_i$ is zero.
Set this variable to true if $w^+_i \neq 0$, and false otherwise.
Any assignment to the variables of $\phi$ must fail to satisfy at least $\alpha m$ clauses, therefore this assignment to the variables of $\varphi$ must fail to satisfy at least $\alpha m - \gamma n$ clauses.

Suppose that clause $j$ appears in $\varphi$ and is not satisfied by the true-false assignment corresponding to $w$.
Then clause $j$ contains either no true variables or at least two true variables.
In the former case, $(y_5)_j = -\ep v_j < - 5\ep /6$.
In the latter, $(y_5)_j > 4 \ep/3 - \ep v_j > \ep/6$.
We obtain in either case that $(y_5)_j^2 > \ep^2/36$.
Summing over the unsatisfied clauses in $\varphi$ yields~$\|y_5\|^2 > (\ep^2/36) (\alpha m - \gamma  n)$.

We obtain
\begin{align*}
\|\tilde Xw\|^2 & \geq \|y_1\|^2 + \|y_2\|^2 + \|y_5\|^2 \\
& > (1 - \beta )n + (\ep^2/36) (\alpha m - \gamma  n) \\
& = (1+\rho) n\,.
\end{align*}

Since $w$ was a minimizer of~\eqref{eqn:the_program}, Property~3 holds with $c_2 = 1+\xi^2$ and $c_3 = \rho$, as claimed.
\end{proof}

\subsection{Verification of conditions of Proposition~\ref{prop:approximate_assignment_implies_bound}}\label{sec:verify}
In order to verify the conditions of Proposition~\ref{prop:approximate_assignment_implies_bound}, we require several lemmas about the vector $w$.
Lemma~\ref{prop:low_variance} establishes that both $v$ and $w^+ + w^-$ are close to multiples of $\1$.
\begin{lemma}\label{prop:low_variance}
Let $\overline{v} = \frac 1 n \1^\top v$.
The following bounds hold:
\begin{align*}
\sum_{i=1}^n (v_i - \overline{v})^2 & < 2 \xi^2 n\,, \\
\sum_{i=1}^n (w^+_i + w^-_i - \overline v)^2 & < 8 \xi^2 n\,, \\
\overline{v}^2 & > 1 - 3\ep^2\,.
\end{align*}
\end{lemma}

Lemma~\ref{prop:mostly_single} establishes that for most $i \in [n]$, exactly one of $w^+_i$ and $w_i^-$ is nonzero.
\begin{lemma}\label{prop:mostly_single}
Let
\begin{align*}
I & = \{i: w^+_i \neq 0, w^-_i \neq 0\} \\
J & = \{j: w^+_j = 0, w^-_j = 0\}\,.
\end{align*}
If $\ep^2 < 1/6$, then
\begin{equation*}
|I| + |J| < 38 \xi^2 n\,.
\end{equation*}
\end{lemma}

Proofs of both lemmas appear in the Appendix.

With these lemmas in hand, we now show that the two conditions of Proposition~\ref{prop:approximate_assignment_implies_bound} are satisfied for appropriate choices of $\beta$ and $\gamma$.
We first show that $\|w^+\|^2 + \|w^-\|^2 \approx n$.
The proof is based on a simple observation: Lemma~\ref{prop:low_variance} shows that $w^+ + w^-$ and $v$ are close, and Lemma~\ref{prop:mostly_single} shows that $w^+$ and $w^-$ have almost disjoint support.
Together, these two facts imply that $2 (\|w^+\|^2 + \|w^-\|^2) \approx \|w^+\|^2 + \|w^-\|^2 + \|w^+ + w^-\|^2 \approx \|w^+\|^2 + \|w^-\|^2 + \|v\|^2 = 2n$.
\begin{prop}\label{prop:large_top}
If $\ep^2 < 1/6$, then
\begin{equation*}
\|w^+\|^2 + \|w^-\|^2 > (1- 25\xi)n\,.
\end{equation*}
\end{prop}
\begin{proof}
As in Lemma~\ref{prop:mostly_single}, let
\begin{align*}
I & = \{i: w^+_i \neq 0, w^-_i \neq 0\} \\
J & = \{j: w^+_j = 0, w^-_j = 0\}\,.
\end{align*}
Let $S = [n] \setminus (I \cup J)$.
If $i \in S$, then exactly one of $w^+_i$ and $w^-_i$ is nonzero, so $\sum_{i \in S} (w^+_i + w^-_i)^2 = \sum_{i \in S} (w^+_i)^2 + (w^-_i)^2 \leq \|w^+\|^2 + \|w^-\|^2$.

Expanding the square yields
\begin{align*}
(w_i^+ + w_i^-)^2 & = \overline{v}^2 + (w_i^+ + w_i^- - \overline{v})^2 + 2 \overline{v} (w_i^+ + w_i^- - \overline{v})\\ 
&\geq \overline{v}^2 - 2 \overline{v} \cdot |w_i^+ + w_i^-- \overline{v}|\,,
\end{align*}
whence the lower bound
\begin{align}
\|w^+\|^2 + \|w^-\|^2
& \geq \sum_{i \in S} (w_i^+ + w_i^-)^2 \nonumber \\
& \geq |S| \overline{v}^2 - 2 \overline{v} \sum_{i \in S} |w_i^+ - w_i^-- \overline{v}| \nonumber \\
& > |S| \overline{v}^2 - 4 \sqrt 2 \xi n \overline{v}\,, \label{eq:mass_bound_1}
\end{align}
where the third inequality follows from the Cauchy-Schwarz inequality and Lemma~\ref{prop:low_variance}.
On the other hand, Lemma~\ref{prop:low_variance} and the fact that $\|w^+\|^2 +  \|w^-\|^2 + \|v\|^2 = 2n$ imply
\begin{align}
\|w^+\|^2 + \|w^-\|^2 = 2n - \|v\|^2 & = 2n - n \overline{v}^2 - \sum_{i=1}^n (v_i - \overline v)^2 \nonumber \\
& > 2n - n \overline{v}^2 - 2 \xi^2 n\,. \label{eq:mass_bound_2}
\end{align}
Summing~\eqref{eq:mass_bound_1} and \eqref{eq:mass_bound_2} and using the fact that $|S| > n - 38 \xi^2 n$ by Lemma~\ref{prop:mostly_single} yields
\begin{equation*}
\|w^+\|^2 + \|w^-\|^2 \geq n - 19 \xi^2 n \overline{v}^2 - \xi^2 n - 2 \sqrt 2 \xi n \overline v\,.
\end{equation*}
Finally, we apply the fact that $n \overline{v}^2 \leq \|v\|^2 \leq 2n$ and the assumption that $\xi < \ep < 1/2$ to conclude
\begin{equation*}
\|w^+\|^2 + \|w^-\|^2 \geq n - 39 \xi^2 n - 4 \xi n > n - 25 \xi n.
\end{equation*}
\end{proof}

Finally, we show that most clauses are good.
Again, the proof follows from an easy calculation: we show that $\overline v$ is close to 1, and Lemma~\ref{prop:low_variance} implies that $w_i^+ + w_i^-$ and $v_i$ are both close to $\overline v$ for all $i \in [n]$.
\begin{prop}\label{prop:few_bad_clauses}
If $\ep^2 \leq 1/25$ and $\xi \leq 1/200$, then there are at most $1284\xi^2n$ bad clauses.
\end{prop}
\begin{proof}
Recall that, by assumption, $\overline v \geq 0$.
We first show
\begin{equation}\label{eq:vbar_near_one}
|\overline v - 1| < 1/12\,.
\end{equation}

Lemma~\ref{prop:low_variance} implies $\overline{v}^2 > 1 - 3\ep^2 \geq 22/25$, so $\overline{v} > 11/12$.
By convexity, $n \overline{v}^2 \leq \|v\|^2$, and Proposition~\ref{prop:large_top} implies that $\|v\|^2 \leq (1+25\xi)n$.
Therefore $\overline{v}^2 \leq 1 + 25 \xi \leq 9/8$ and $\overline{v} < 13/12$.

We now show that most entries of $v$ and $w^+ + w^-$ are near $1$.
Let
\begin{align*}
L & = \{\ell: |v_\ell - 1| > 1/6\}\,, \\
K & = \{k: |w_k^+ + w_k^- - 1| > 1/3\}\,.
\end{align*}

If $\ell \in L$, then~\eqref{eq:vbar_near_one} implies $(v_\ell - \overline{v})^2 \geq (|v_\ell-1| - |\overline v - 1|)^2 > 1/144$.
Likewise, if $k \in K$, then $(w_k^+ + w_k^- - \overline{v})^2 \geq (|w_k^+ + w_k^- - 1| - |\overline v - 1|)^2 > 1/16$.
Summing these inequalities over $\ell \in L$ and $k \in K$, respectively, and applying Lemma~\ref{prop:low_variance} yields
\begin{align}
\label{eq:lsmall}|L| & < 144 \sum_{i=1}^n (v_i - \overline{v})^2 < 288 \xi^2 n\,, \\
\label{eq:ksmall}|K| &< 16 \sum_{i=1}^n (w_i^+ - w_i^- - \overline{v})^2 < 128 \xi^2 n\,.
\end{align}

By Lemma~\ref{prop:mostly_single} and~\eqref{eq:ksmall}, there exists a set $G \in [n]$ of size at least $(1-166\xi^2)n$ such that for all $i \in G$:
\begin{itemize}
\item Exactly one of $w^+_i$ and $w^-_i$ is nonzero, and
\item $|w^+_i + w^-_i - 1| < 1/3$.
\end{itemize}
By definition, all variables in $G$ are good.

Since $\phi$ is $6$-bounded, the other $166 \xi^2 n$ variables are contained in at most $996 \xi^2 n$ clauses.
This fact combined with~\eqref{eq:lsmall} implies that there are at most $(996 \xi^2 + 288 \xi^2)n = 1284 \xi^2 n$ bad clauses, as claimed.
\end{proof}

We are now ready to prove Theorem~\ref{thm:main} for a specific choice of $\delta$.
\begin{prop}\label{prop:main}
Set $\ep = 1/5$ and $\xi = 1/\lceil 10^{5}/\alpha\rceil$, and let
\begin{equation*}
\rho = \frac{\ep^2}{36} \left(\frac{9}{13} \alpha - 1284 \xi^2\right) - 25 \xi\,.
\end{equation*}
If $\tilde X$ is defined as in Section~\ref{sec:X_def}, then $\tilde X$ satisfies the three properties given in Section~\ref{sec:overview}, with $c_1 = 3 \xi^{-1}$, $c_2 = 1 + \xi^2$, and $c_3 = \rho$.
\end{prop}
\begin{proof}
Properties~1 and~2 have been shown to hold in Propositions~\ref{prop:property_1} and~\ref{prop:property_3}.
By Propositions~\ref{prop:large_top} and~\ref{prop:few_bad_clauses}, if $\ep^2 \leq 1/25$ and $\xi \leq 1/200$, then the conditions of Proposition~\ref{prop:approximate_assignment_implies_bound} hold with $\beta = 25\xi$ and $\gamma = 1284 \xi^2$.
Plugging in the given values of $\ep$ and $\xi$ yields that $\rho > 0$, so Property~3 holds with $c_2 = 1 + \xi^2$ and $c_3 = \rho$.
\end{proof}
Proposition~\ref{prop:main} implies that any procedure to distinguish between $\frac{1}{c_1}\tilde X \in \rip(2c_2n, 1 - \frac{1 + c_3}{2c_1^2})$ and $\frac{1}{c_1}\tilde X \notin \rip(2n, 1 - \frac{1}{2c_1^2 - c_3})$ would yield a procedure to distinguish between $\val(\phi) = 1$ and $\val(\phi) < 1 - \alpha$, so Theorem~\ref{thm:main} holds with
\begin{align*}
\delta  = 1 - \frac{1 + \rho}{18 \xi^{-2}}, \qquad
\lambda_1  = c_2, \qquad
\lambda_2  = \frac{18 \xi^{-2}}{18 \xi^{-2} - \rho}\,.
\end{align*}

\subsection{Extension to general $\delta$}\label{sec:extension}
Finally, Theorem~\ref{thm:main} follows from the following proposition.
\begin{prop}\label{prop:new_delta}
If Theorem~\ref{thm:main} holds for some $\delta \in (0, 1)$, then it holds for \emph{any} $\delta' \in (0, 1)$.
\end{prop}
\begin{proof}
In both cases, we proceed by showing how to reduce the problem of distinguishing $X \in \rip(k, \delta)$ from $X \notin \rip(k/\lambda_1, \lambda_2 \delta)$ to the problem of distinguishing $X' \in \rip(k, \delta')$ from $X' \notin \rip(k/\lambda_1, \lambda_2' \delta')$ for some constant $\lambda_2' > 1$, where the matrix $X'$ satisfies $\|X' u\|^2 \leq (1+\delta')\|u\|^2$ for all $u$ and can be constructed from $X$ in polynomial time.

We first suppose $\delta' \in (0, \delta)$.
Given a matrix $X \in \RR^{n \times p}$, define the block matrix $X' \in \RR^{(n + p) \times p}$ by
\begin{equation*}
X' = \left(\begin{array}{c}
\mu X \\
\nu I_{p \times p}
\end{array}\right)\,,
\end{equation*}
where $I_{p \times p}$ is the identity matrix and $\mu$ and $\nu$ are rational approximations of $\sqrt{\delta'/\delta}$ and $\sqrt{1 - (\delta'/\delta)}$, respectively, such that $\mu^2 \delta \in [\delta' - 2 \tau, \delta' - \tau]$ and $\mu^2 + \nu^2 \in [1 - \tau, 1]$ for some tolerance parameter $\tau \leq \frac{(\lambda_2 - 1)\delta'}{2 + 4 \lambda_2}$.
Note that $\mu$ and $\nu$ are constants independent of the problem instance, and the matrix $X'$ can be constructed from $X$ in polynomial time.

The definition of $X'$ implies that for any $u \in \RR^p$, if $\|X u\|^2 = (1 + \theta)\|u\|^2$, then
\begin{equation*}
\|X' u\|^2 = \mu^2 \|X u\|^2 + \nu^2 \|u\|^2 = ((\mu^2 + \nu^2) + \mu^2 \theta)\|u\|^2\,.
\end{equation*}
Note that $\|X'u\|^2 \leq (1+\delta') \|u\|^2$ for all $u$.
Moreover, it is easy to check if $X \in \rip(k, \delta)$, then $X' \in \rip(k, \delta')$, and if $X \notin \rip(k/\lambda_1, \lambda_2 \delta)$, then $X' \notin \rip(k/\lambda_1, \lambda_2 \delta' - (1+ 2\lambda_2)\tau)$, in which case the definition of $\tau$ implies that $X' \notin \rip(k/\lambda_1, \lambda_2' \delta')$ for $\lambda_2' = \frac 1 2 (\lambda_2 + 1)$.
Therefore the problem of distinguishing between $X \in \rip(k, \delta)$ and $X \notin \rip(k/\lambda_1, \lambda_2 \delta)$ is reducible in polynomial time to the problem of distinguishing between $X' \in \rip(k, \delta')$ or $X' \notin \rip(k/\lambda_1, \lambda_2'\delta')$.

For $\delta' \in (\delta, 1)$, we proceed similarly.
Let $X$ be a matrix satisfying $\|Xu\|^2 \leq (1+\delta)\|u\|^2$ for all $u$.
Given such a matrix, let $X' = \mu X$, where $\mu \leq 1$ is a rational number such that $\mu^2 \in [\frac{1-\delta'}{1-\delta}, \frac{1- \delta'}{1 - \lambda_2 \delta})$\,.
As above, $X'$ can be constructed from $X$ in polynomial time.
Note that $\|X'u\|^2 \leq (1+\delta)\|u\|^2 \leq (1+ \delta')\|u\|^2$ for all $u$.

If $X \in \rip(k, \delta)$, then $X' \in \rip(k, \delta')$.
If $X \not\in \rip(k/\lambda_1, \lambda_2 \delta)$, then $X' \not\in \rip(k/\lambda_1, \lambda_2' \delta')$ for $\lambda_2 = \frac{1 - \mu^2(1 - \lambda_2 \delta)}{\delta'} > 1$.
Hence deciding whether $X \in \rip(k, \delta)$ or $X \not\in \rip(k/\lambda_1, \lambda_2 \delta)$ is reducible in polynomial time to deciding whether $X' \in \rip(k, \delta')$ or $X' \not\in \rip(k/\lambda_1, \lambda_2' \delta')$.
\end{proof}

\subsection{Extension to matrices of other shapes}\label{sec:shape}
As noted above, the matrices constructed in the proof of Theorem~\ref{thm:main} have more rows than columns.
We first show that Problem~\ref{prob:modified_no_gap} remains hard when restricted to square matrices.
\begin{prop}\label{prop:square}
Theorem~\ref{thm:main} holds when restricted to square matrices.
\end{prop}
\begin{proof}
It suffices to show that for any tolerance parameter $\tau$ and matrix $X \in \RR^{n \times p}$ with $n \geq p$  and $\|X\|_\text{op} \leq 2$ it is possible to construct in polynomial time a square matrix $\hat X \in \RR^{p \times p}$ such that
\begin{equation}\label{eqn:qr_guarantee}
\big|\|Xu\|^2 - \|\hat X u\|^2\big| \leq \tau\|u\|^2\,.
\end{equation}
First, let us see why this implies the proposition.
Fix $\delta \in (0, 1)$.
Choose a $\delta' < \delta$ and let $\lambda'_1$ and $\lambda'_2$ be the constants corresponding to $\delta'$ as in the statement of Theorem~\ref{thm:main}.
Set $\tau \leq \min\{\frac 1 4 (\lambda'_2 - 1)\delta', \delta - \delta'\}$.
Given any matrix $X$ with more rows than columns, we can construct a square matrix $\hat X$ satisfying~\eqref{eqn:qr_guarantee} in polynomial time.
Recall that we can assume that $\|Xu\|^2 \leq (1+\delta')\|u\|^2$.
If $X \in \rip(k, \delta')$, then $\hat X \in \rip(k, \delta'')$ for $\delta'' = \delta' + \tau \leq \delta$, and if $X \notin \rip(k/\lambda'_1, \lambda'_2 \delta')$, then $\hat X \notin \rip(k/\lambda_1', \lambda_2'' \delta'')$ for $\lambda_2'' = \frac{\lambda_2' \delta' - \tau}{\delta' + \tau} > 1$.
Moreover, since $\|Xu\|^2 \leq (1+\delta')\|u\|^2$, we likewise have $\|\hat Xu\|^2 \leq (1+\delta'+\tau)\|u\|^2 = (1+\delta'')\|u\|^2$.
Since it is \np-hard to distinguish between $X \in \rip(k, \delta')$ and $X \notin \rip(k/\lambda_1', \lambda_2' \delta')$, it is \np-hard to distinguish between $\hat X \in \rip(k, \delta'')$ and $\hat X \notin (k/\lambda_1', \lambda_2'' \delta'')$ for square matrices, so the claim holds for some $\delta'' \leq \delta$.
To conclude, we apply the reduction in the second half of Proposition~\ref{prop:new_delta}, which does not change the shape of the matrix, to show that if the claim holds for $\delta''$, then it holds for $\delta$.

We now show that we can indeed find matrices satisfying~\eqref{eqn:qr_guarantee} in polynomial time.
Given any matrix $X \in \RR^{n \times p}$ with $\|X\|_\text{op} \leq 2$ and desired accuracy threshold $\tau$, the Householder QR algorithm (see~\cite[Section 5.2.2]{Gol13}) computes in polynomial time an upper triangular matrix $R \in \RR^{n \times p}$ such that
\begin{equation}\label{eqn:op_norm_guarantee}
\|Q R - X\|_{\text{op}} \leq \frac{\tau}{4}
\end{equation}
for some orthogonal matrix $Q \in \RR^{n \times n}$.
Note that if $X$ has more rows than columns, then we can write
\begin{equation*}
R = \left(\begin{array}{c}
\hat X \\
0
\end{array}\right)\,,
\end{equation*}
where $\hat X \in \RR^{p \times p}$ is upper triangular and $0$ is the $(n - p) \times p$ zeroes matrix.
For any $u \in \RR^p$, we have $\|Q R u\| = \|R u\| = \|\hat X u\|$, and combining this with~\eqref{eqn:op_norm_guarantee} gives
\begin{equation*}
\big|\|Xu\|^2 - \|\hat X u\|^2\big| = \big|\|Xu\| - \|\hat X u\|\big|(\|Xu\| + \|\hat X u\|) \leq \tau \|u\|^2\,,
\end{equation*}
as desired.
\end{proof}

We conclude by noting that we can easily extend to rectangular matrices with more columns than rows by combining the square matrices produced in Proposition~\ref{prop:square} with rectangular matrices in $\rip(k, \delta)$, an observation which appears in~\cite{KoiZou14}.
We rely on the following simple lemma.
\begin{lemma}[Koiran and Zouzias {\cite{KoiZou14}}]\label{lem:block_diagonal}
Let $X$ be a block diagonal matrix
\begin{equation*}
X = \left(\begin{array}{ll}
A & 0 \\
0 & B
\end{array}\right)\,,
\end{equation*}
where $A$ and $B$ are matrices with at least $k$ columns and $0$ represents a zeroes matrix of the appropriate size.
Then $X \in \rip(k, \delta)$ iff $A, B \in \rip(k, \delta)$.
\end{lemma}
Combining this lemma with a deterministic procedure for generating rectangular \rip\ matrices establishes that Problem~\ref{prob:modified_no_gap} is still hard for matrices with many more columns than rows.

\begin{prop}\label{prop:rectangular}
For all $\delta \in (0, 1)$ there exist constants $\lambda_1, \lambda_2 > 1$ such that for any constant $c > 1$, given a matrix $X \in \RR^{n \times p}$ with $p \geq c n$ and sparsity parameter $k$, it is \np-hard to decide whether $X \in \rip(k, \delta)$ or $X \notin \rip(k/\lambda_1, \lambda_2 \delta)$.
\end{prop}
\begin{proof}
We reduce from the square case of Proposition~\ref{prop:square}.
Fix $\delta \in (0, 1)$.
There exists a constant $\ep$ such that, for $r$ sufficiently large, the deterministic construction of~\cite{BouDilFor11} can in polynomial time generate a matrix $B \in \RR^{r^2 \times r^{2 + \ep}}$ such that $B \in \rip(r, \delta)$.
Given a square matrix $A \in \RR^{r \times r}$, construct the block matrix $X \in \RR^{(r + r^2) \times (r + r^{2 + \ep})}$ from $A$ and $B$ as in Lemma~\ref{lem:block_diagonal}.
For any $k \leq r$, we have $B \in \rip(r, \delta) \subseteq \rip(k, \delta)$.
Therefore, if $A \in \rip(k, \delta)$, then $X \in \rip(k, \delta)$, and if $A \notin \rip(k/\lambda_1, \lambda_2 \delta)$ then $X \notin \rip(k/\lambda_1, \lambda_2 \delta)$.
By restricting our attention to $r$ large enough that $r^\ep \geq 2c$, we obtain the claim.
\end{proof}

\section{Conclusion}
In this work, we show that it is \np-hard to certify the Restricted Isometry Property, even approximately, for all $\delta \in (0, 1)$.
This resolves a question implicit in earlier work, which either required $\delta = 1 - o(1)$ or relied on stronger assumptions than $\textsc{p} \neq \np$.
Our proof proceeded via a reduction from $6$-bounded max positive 1-in-3 \sat, whose hardness we established in Proposition~\ref{prop:max_bounded_monotone_hardness}.
While similar harness results exist elsewhere in the literature, this bounded variant may be of use in other reductions.

We note that we have made no attempt to optimize the constants in the proof of Theorem~\ref{thm:main}, but even a more careful version of this proof will still produce $\lambda_1$ and $\lambda_2$ very close to $1$.
It is an open question whether \np-hardness can be proven for a version of Problem~\ref{prob:modified_no_gap} in which the constants $\lambda_1$ and $\lambda_2$ are large.

The most important open question in this area is to establish how difficult certifying $X \in \rip(k, \delta)$ is \emph{on average} when $X$ is drawn from some natural probability distribution.
As noted above, the only work we know of to focus on this question is~\cite{WanBerPla16}.
Extending their results to the full range of parameters considered in practice would be an important theoretical result.

\appendix
\section{Proof of Proposition~\ref{prop:max_bounded_monotone_hardness}}
We reduce from a problem called max 3\sat-5.
An instance of max 3\sat-5 is a \textsc{cnf} formula where each clause contains exactly 3 variables (in positive or negative form) and each variable appears in exactly 5 clauses.
The max 3\sat-5 problem is known to be gap hard:

\begin{prop}[Feige~{\cite{Fei98}}]
There exists a constant $\alpha'$ such that, given an instance $\psi$ of max 3\sat-5, it is \np-hard to distinguish between $\val(\psi) = 1$ and $\val(\psi) < (1-\alpha')$.
\end{prop}

Given an instance $\psi$ of max 3\sat-5 with $n$ variables and $m = \frac{5n}{3}$ clauses, we produce an instance $\phi$ of $6$-bounded max positive 1-in-3 \sat\ with $n' = 2n+4m$ variables and $m' = 3m + n = \frac{9}{13} n'$ clauses such that:
\begin{itemize}
\item $\val(\psi) = 1 \implies \val(\phi) = 1$,
\item $\val(\psi) < (1-\alpha') \implies \val(\phi) < (1-\alpha)$, where $\alpha = \alpha' / 18 $.
\end{itemize}
The claimed \np-hardness of distinguishing $\val(\phi) = 1$ and $\val(\phi) < (1-\alpha)$ then follows.

\begin{proof}[Proof of Proposition~\ref{prop:max_bounded_monotone_hardness}]
We will produce the instance $\phi$ from $\psi$ in several stages, first by transforming $\psi$ into an instance $\psi'$ of 1-in-3 \sat\ that contains negated variables, and then transforming $\psi'$ into an instance $\phi$ of \emph{positive} 1-in-3 \sat.

We first produce an instance of 1-in-3 \sat\ that contains negated variables.
Consider a clause $(a \vee b \vee c) \in \psi$, where $a, b, c$ represent arbitrary literals (positive or negative variables).
Replace this clause by the three 1-in-3 \sat\ clauses:
\begin{align*}
\exac(a, z_1, z_2),\, \exac(\bar b, z_1, z_3),\, \exac(\bar c, z_2, z_4)\,,
\end{align*}
where $\exac$ denotes the ``exactly one'' predicate as in Definition~\ref{def:e1}, $\bar x$ denotes the negated version of the literal $x$, and $z_1, \dots, z_4$ are four fresh variables appearing in these three clauses and no others.
If $(a \vee b \vee c)$ is not satisfied, then at least one of the new clauses is unsatisfied.
Indeed, in that case $\bar b$ and $\bar c$ are true, so at least one of $\exac(\bar b, z_1, z_3)$ and $\exac(\bar c, z_2, z_4)$ is unsatisfied unless $z_1, \dots, z_4$ are all false, in which case $\exac(a, z_1, z_2)$ is unsatisfied.
On the other hand, if $(a \vee b \vee c)$ is satisfied, then it is easy to check that there is a setting of $z_1, \dots, z_4$ to satisfy all three of the new clauses.
Repeating this replacement for each clause in $\psi$ yields $\psi'$

To obtain a positive instance, replace each occurrence of $x_i$ or $\overline{x_i}$ by the new variable $w_i$ or $y_i$, respectively, and add the clause
\begin{equation*}
\exac(w_i, y_i)\,.
\end{equation*}
Call the resulting positive 1-in-3 \sat\ instance $\phi$.
Note that $\phi$ has $m' = 3m + n$ clauses, and that each variable appears in at most $6$ clauses.
The instance $\phi$ involves $n' = 2n + 4m$ variables, of which the $2n$ variables $w_1, \dots w_n, y_1, \dots y_n$ correspond to positive and negative versions of $\{x_1, \dots, x_n\}$ in $\psi$.

If $\psi$ is satisfiable, then clearly $\phi$ is as well.
We now prove the contrapositive of the second claim: if $\val(\phi) \geq 1- \alpha$, then $\val(\psi) \geq 1- \alpha'$.
Suppose an assignment to the variables in $\phi$ leaves at most $\alpha m'$ clauses unsatisfied.
We use this assignment to obtain an assignment to the variables in $\psi$ in the following way: if only one of the variables $w_i$ or $y_i$ is true, then let $x_i$ be true if $w_i$ is true and let $x_i$ be false if $y_i$ is true.
If $w_i$ and $y_i$ are both true or both false, then $\phi$ does not determine a value of $x_i$, so we set it to false arbitrarily.

The clauses in $\phi$ are of two types: those that correspond to clauses in $\psi$ and those of the form $\exac(w_i, y_i)$ corresponding to variables in $\psi$.
Each unsatisfied clause of the first type corresponds to at most one unsatisfied clause of $\psi$, and each unsatisfied clause of the form $\exac(w_i, y_i)$ corresponds to at most five unsatisfied clauses of $\psi$, since the variable $x_i$ appears in five clauses of $\psi$.
In either case, an unsatisfied clause in $\phi$ induces at most five unsatisfied clauses in $\psi$.
Since the assignment to the variables of $\phi$ satisfied all but at most $\alpha m'$ clauses, the corresponding assignment to $\psi$ has at most $5 \alpha m'$ unsatisfied clauses.

Hence
\begin{equation*}
\val(\psi) \geq \frac{m - 5 \alpha m'}{m} = 1 - 18 \alpha = 1- \alpha'\,,
\end{equation*}
as desired.
\end{proof}

\section{Proofs of Lemmas~\ref{prop:low_variance} and~\ref{prop:mostly_single}}
\begingroup
\def\thelemma{\ref{prop:low_variance}}
\begin{lemma}
Let $\overline{v} = \frac 1 n \1^\top v$.
The following bounds hold:
\begin{align}
\label{eq:low_variance_v}\sum_{i=1}^n (v_i - \overline{v})^2 & < 2 \xi^2 n\,, \\
\label{eq:low_variance_w}\sum_{i=1}^n (w^+_i + w^-_i - \overline v)^2 & < 8 \xi^2 n\,, \\
\label{eq:vbar_large}\overline{v}^2 & > 1 - 3\ep^2\,.
\end{align}
\end{lemma}
\endgroup
\begin{proof}
We first show a simple upper bound on $\|\tilde X w\|^2$.
Let $u = (0, \dots, 0, 1, \dots, 1)$ be the assignment vector corresponding to the assignment that sets each variable to false.
By Proposition~\ref{prop:assignment_value}\,,
\begin{equation*}
\|\tilde Xu\|^2 = n + \ep^2 m \leq (1+\ep^2)n\,.
\end{equation*}
Since $w$ satisfies~\eqref{eqn:the_program},
\begin{equation}\label{eq:easy_bound}
\|\tilde Xw\| \leq \|\tilde Xu\|^2 \leq (1+\ep^2)n < 2 n\,,
\end{equation}
where we have used the assumption that $\ep < 1$.

By definition,
\begin{equation*}
\|P v\|^2 = \|(I - \frac 1 n \1 \1^\top) v\|^2 = \sum_{i=1}^n (v_i - \overline{v})^2\,.
\end{equation*}
By~\eqref{eq:easy_bound},
\begin{equation*}
\|\xi^{-1} P v\|^2 \leq \|\tilde X w\|^2 < 2n\,,
\end{equation*}
hence
\begin{equation*}
\sum_{i=1}^n (v_i - \overline{v})^2 < 2 \xi^2 n\,,
\end{equation*}
as claimed.

For the second bound, by Young's inequality,
\begin{equation*}
\sum_{i=1}^n (w^+_i + w^-_i - \overline v)^2 \leq 2 \sum_{i=1}^n (w^+_i - w^-_i - v_i)^2 + 2 \sum_{i=1}^n (v_i - \overline v)^2\,.
\end{equation*}
Note that
\begin{align*}
\xi^{-2} \sum_{i=1}^n (w^+_i - w^-_i - v_i)^2 & = \|\xi^{-1}I w^+ + \xi^{-1} I w^- - \xi^{-1}I v\|^2 \\
&\leq \|\tilde X w\|^2\,.
\end{align*}
By~\eqref{eq:easy_bound}, this quantity is smaller than $2n$.
Combining this with~\eqref{eq:low_variance_v} yields
\begin{equation*}
2 \sum_{i=1}^n (w^+_i - w^-_i - v_i)^2 + 2 \sum_{i=1}^n (v_i - \overline v)^2 < 8 \xi^2 n\,,
\end{equation*}
and~\eqref{eq:low_variance_w} follows.

For the third inequality, by~\eqref{eq:easy_bound},
\begin{equation*}
\|w^+\|^2 + \|w^-\|^2 \leq \|\tilde X w\|^2 \leq (1+\ep^2)n\,.
\end{equation*}
Therefore
\begin{equation*}
\|v\|^2 = 2n - (\|w^+\|^2 + \|w^-\|^2) \geq (1-\ep^2)n\,.
\end{equation*}
By~\eqref{eq:low_variance_v}\,,
\begin{equation*}
\|v\|^2 = \sum_{i=1}^n (v_i - \overline{v})^2 + n \overline{v}^2 < 2 \xi^2 n +  n \overline{v}^2\,.
\end{equation*}
We obtain
\begin{equation*}
n\overline{v}^2 > \|v\|^2 - 2 \xi^2 n > (1-\ep^2)n - 2\ep^2 n = (1-3 \ep^2)n\,,
\end{equation*}
and the claim follows.
\end{proof}
\begingroup
\def\thelemma{\ref{prop:mostly_single}}
\begin{lemma}
Let
\begin{align*}
I & = \{i: w^+_i \neq 0, w^-_i \neq 0\} \\
J & = \{j: w^+_j = 0, w^-_j = 0\}\,.
\end{align*}
If $\ep^2 < 1/6$, then
\begin{equation*}
|I| + |J| < 38 \xi^2 n\,.
\end{equation*}
\end{lemma}
\endgroup
\begin{proof}
If $w^+_j$ and $w^-_j$ are both $0$ for some $j$, then
\begin{equation*}
(w^+_j + w^-_j - \overline{v})^2 = \overline{v}^2 > (1-3\ep^2) > 1/2
\end{equation*}
by Lemma~\ref{prop:low_variance}.
Summing both sides of the above inequality over $j \in J$ yields
\begin{equation*}
|J|/2 < \sum_{j \in J} (w^+_j + w^-_j - \overline{v})^2 \leq \sum_{j = 1}^n (w^+_j + w^-_j - \overline{v})^2\,.
\end{equation*}
Applying~\eqref{eq:low_variance_w} yields
\begin{equation*}
|J| < 16 \xi^2 n\,.
\end{equation*}

We now show that $v$ has almost full support.
If $v_i = 0$, then by Lemma~\ref{prop:low_variance},
$(v_i - \overline{v})^2 = \overline{v}^2 > (1 - 3 \ep^2) > 1/2$.
If $p$ is the number of zero entries in $v$, then summing this inequality yields
\begin{equation*}
p/2 < \sum_{i: v_i = 0} (v_i - \overline v)^2 \leq \sum_{i = 1}^n (v_i - \overline v)^2\,.
\end{equation*}
Applying~\eqref{eq:low_variance_v} then implies $p \leq 4 \xi^2 n$, so
\begin{equation*}
\|v\|_0 = n - p > (1 - 4\xi^2)n\,.
\end{equation*}
Since $\|w\|_0 \leq 2 c_2 n = 2n + 2 \xi^2 n$, we have
\begin{equation*}
\|w^+\|_0 + \|w^-\|_0 = \|w\|_0 - \|v\|_0 < (1 + 6 \xi^2)n\,.
\end{equation*}
We obtain
\begin{align*}
|I| &= \|w^+\|_0 + \|w^-\|_0 + |J| - n \\ &< (1+ 6 \xi^2)n + 16 \xi^2 n - n = 22 \xi^2 n\,.
\end{align*}
Combining the above bounds on $|I|$ and $|J|$ yields the claim.
\end{proof}
\bibliographystyle{plain}
\bibliography{rip}

\end{document}